\documentclass[aps,prl,twocolumn,superscriptaddress,floatfix,nofootinbib,showpacs,longbibliography,groupedaddress]{revtex4-2}

\usepackage[utf8]{inputenc}  
\usepackage[T1]{fontenc}     
\usepackage[british]{babel}  
\usepackage[sc,osf]{mathpazo}
\usepackage[table]{xcolor}
\usepackage[scaled=0.86]{berasans}  
\usepackage[colorlinks=true, citecolor=blue, urlcolor=blue]{hyperref}
\makeatletter
\newcommand{\setword}[2]{%
  \phantomsection
  #1\def\@currentlabel{\unexpanded{#1}}\label{#2}%
}
\makeatother
\usepackage{graphicx} 
\usepackage[babel]{microtype}
\usepackage{blkarray}
\usepackage{amsmath}

\usepackage{amsmath,amssymb,amsthm,bm,amsfonts,mathrsfs,bbm} 

\usepackage{xspace}  
\usepackage{pgf,tikz}
\usepackage{xcolor}
\usepackage{multirow}
\usepackage{array}
\usepackage{bigstrut}
\usepackage{braket}
\usepackage{color}
\usepackage{natbib}

\usepackage{multirow}
\usepackage{mathtools}
\usepackage{float}
\usepackage[caption = false]{subfig}
\usepackage{xcolor,colortbl}
\usepackage{color}

\newcommand{\be}{\begin{equation}}
\newcommand{\ee}{\end{equation}}
\newcommand{\ba}{\begin{eqnarray}}
\newcommand{\ea}{\end{eqnarray}}

\newtheorem{theorem}{Theorem}

\newtheorem{proposition}{Proposition}

\newtheorem{lemma}{Lemma}

\def\>{\rangle}
\def\<{\langle}

\usepackage{stmaryrd}

\usepackage{centernot}
\usepackage{subfig}

\usepackage{diagbox}
\usepackage{multirow}
\usepackage{tabularx}

\begin{document}

\title{Stronger Nonlocality in GHZ States: A Step Beyond the Conventional GHZ Paradox}

\author{Ananya Chakraborty}
\affiliation{Department of Physics of Complex Systems, S. N. Bose National Center for Basic Sciences, Block JD, Sector III, Salt Lake, Kolkata 700106, India.}

\author{Kunika Agarwal}
\affiliation{Department of Physics of Complex Systems, S. N. Bose National Center for Basic Sciences, Block JD, Sector III, Salt Lake, Kolkata 700106, India.}

\author{Sahil Gopalkrishna Naik}
\affiliation{Department of Physics of Complex Systems, S. N. Bose National Center for Basic Sciences, Block JD, Sector III, Salt Lake, Kolkata 700106, India.}

\author{Manik Banik}
\affiliation{Department of Physics of Complex Systems, S. N. Bose National Center for Basic Sciences, Block JD, Sector III, Salt Lake, Kolkata 700106, India.}

\begin{abstract}
The Greenberger–Horne–Zeilinger (GHZ) paradox, involving quantum systems with three or more subsystems, offers an {\it all-vs-nothing} test of quantum nonlocality. Unlike Bell tests for bipartite systems, which reveal statistical contradictions, the GHZ paradox demonstrates a definitive (i.e. $1$00\%) conflict between local hidden variable theories and quantum mechanics. Given this, how can the claim made in the title be justified? The key lies in recognising that GHZ games are typically played under a predefined promise condition for input distribution. By altering this promise, different GHZ games can be constructed. Here, we introduce a randomized variant of GHZ game, where the promise condition is randomly selected from multiple possibilities and revealed to only one of the parties chosen randomly. We demonstrate that this randomized GHZ paradox can also be perfectly resolved using a GHZ state, revealing a potentially stronger form of nonlocality than the original paradox. The claim of enhanced nonlocality is supported by its operational implications: correlations yielding perfect success in the randomized game offer a greater communication advantage than traditional GHZ correlations in a distributed multi-party communication complexity task.
\end{abstract}

\maketitle	
{\it Introduction.--} J. S. Bell's celebrated 1964 theorem establishes that quantum mechanical predictions are incompatible with local-hidden-variable theories \cite{Bell1964}. While locality posits that physical systems are influenced only by their immediate surroundings, hidden-variables refer to properties not captured by the quantum mechanical wave function but that otherwise determine experimental outcomes. Bell’s theorem, while addressing a foundational question posed by Einstein, Podolsky, and Rosen \cite{Einstein1935} (see also \cite{Bohr1935,Schrodinger1935,Schrodinger1936}), demonstrates that the conjunction of these classical views does not hold in the quantum world \cite{Bell1966, Mermin1993, Gisin2023}. The theorem is typically illustrated through the violation of Bell-type inequalities \cite{Aspect1981, Aspect1982, Aspect1982(1),Zukowski1993,Weihs1998,Hensen2015,Giustina2015,Shalm2015}, which involve the outcome statistics of freely chosen local measurements performed on the parts of a composite quantum system \cite{Brunner2014}.

Bell’s theorem, when framed as collaborative games, often provides a clearer and more intuitive understanding of quantum nonlocality. For instance, gamification of the celebrated CHSH inequality -- named after Clauser, Horne, Shimony, and Holt \cite{Clauser1969} -- leads to the XOR game \cite{Cleve2004}, where two non communicating parties, Alice and Bob, aim to fulfill the winning condition \( a \oplus b = xy \), with \( a, b \in \{0,1\} \) being their local outputs for respective inputs \( x, y \in \{0,1\} \). Under classical strategies, the optimal success rate in the XOR game is \( 3/4 \), while quantum mechanics allows the optimal success \( P_q = 1/2(1 + 1/\sqrt{2}) \) \cite{Cirelson1980}. A particularly intriguing class of Bell games are quantum pseudo-telepathy games, where quantum strategies ensure perfect success, but classical strategies are strictly suboptimal \cite{Brassard1999}. The first example of such a game involving more than two players was proposed by Greenberger, Horne, and Zeilinger (GHZ) \cite{Greenberger1989} (see also \cite{Mermin1990, Greenberger1990,Bouwmeeste1999, Pan2000,Deng2010,Liu2021}). In the three-party scenario, the winning condition reads as \( a \oplus b \oplus c = x \lor y \lor z \), where \( a, b, c \in \{0,1\} \) are the local outputs for the respective inputs \( x, y, z \in \{0,1\} \). Unlike the XOR game, here the inputs are sampled according to a promise condition \( x \oplus y \oplus z = 0 \). We will denote this game as GHZ\(_{\bf E}\), as the promise condition follows an even parity constraint. Classical players can achieve at most a \( 3/4 \) success rate, whereas quantum players, sharing a three-qubit GHZ state \( \ket{G_3} = (\ket{0}^{\otimes3} + \ket{1}^{\otimes3})/\sqrt{2} \), can achieve perfect success. By altering the promise condition to \( x \oplus y \oplus z = 1 \), a variant of the game, called GHZ\(_{\bf O}\), is obtained, where the winning condition becomes \( a \oplus b \oplus c = x \land y \land z \). Notably, in both GHZ\(_{\bf E}\) and GHZ\(_{\bf O}\), the promise condition is known to all players.

The assertion that GHZ states exhibit nonlocality stronger than the GHZ paradox is explored using a randomized variant of the standard GHZ games, termed R$_2$GHZ. In R$_2$GHZ, the Referee randomly selects which promise condition to apply when distributing inputs, with this information revealed to only one player, who is again chosen randomly. Remarkably, the state \( \ket{G_3} \) still enables a perfect strategy for R$_2$GHZ, leading to correlations that demonstrate stronger nonlocality than conventional GHZ games. The claim of stronger nonlocality is substantiated through operational utility of the correlation in communication complexity tasks \cite{Buhrman2010}. Building on the conventional GHZ game and the R$_2$GHZ variant, we introduce two communication complexity tasks, \(\mathrm{CC_2}\) and \(\mathrm{R_2CC_2}\) respectively, involving three distant servers: Alice and Bob as senders, and Charlie as the computing server. Both tasks can be executed with only one bit of communication from Alice and Bob to Charlie, provided they share the state \( \ket{G_3} \). In contrast, without the GHZ state, \(\mathrm{CC_2}\) requires Charlie to share a 1-bit channel with one sender and two 1-bit channels with the other, while \(\mathrm{R_2CC_2}\) requires two 1-bit channels from both senders. Thus, the GHZ state provides a one-bit advantage in \(\mathrm{CC_2}\) and a two-bit advantage in \(\mathrm{R_2CC_2}\), underscoring stronger nonlocality in the latter.

{\it Three-party GHZ game.--} The Referee provides binary inputs \(x, y, z \in \{0,1\}\) to three distant parties -- Alice, Bob, and Charlie -- who return binary outputs \(a, b, c \in \{0,1\}\), respectively. Unlike the CHSH game \cite{Clauser1969}, where inputs are sampled independently and uniformly at random, in this case Referee chooses the inputs uniformly at random while satisfying a specific promise condition, which accordingly fixes the winning criteria. Two variants of the game can be defined as follows:
\begin{subequations}
\begin{align}
&\text{(i) Even Parity GHZ Game [GHZ$_{\bf E}$]~:--}\nonumber\\ 
&~~~~~\left\{\!\begin{aligned}
\text{Promise Condition,}~{\bf P_E}:~x \oplus y \oplus z = 0,\\
\text{Winning Condition:}~a \oplus b \oplus c = x \lor y \lor z
\end{aligned}\right\};\\
&\text{(ii) Odd Parity GHZ Game [GHZ$_{\bf O}$]~:--}\nonumber\\ 
&~~~~~\left\{\!\begin{aligned}
\text{Promise Condition,}~{\bf P_O}:~x \oplus y \oplus z = 1,\\
\text{Winning Condition:}~a \oplus b \oplus c = x \land y \land z
\end{aligned}\right\}.
\end{align}
\end{subequations}
Interestingly, both games can be won perfectly if the players share the state \(\ket{G_3}\). The quantum strategies for GHZ$_{\bf E}$ and GHZ$_{\bf O}$ are detailed in Table \ref{tab1}. Notably, for fixed strategies of Bob and Charlie, perfect success in GHZ$_{\bf E}$ and GHZ$_{\bf O}$ can be achieved by adjusting only Alice's strategy according to the respective promise conditions. This observation will play a crucial role in the subsequent sections.

{\it Advantage in communication complexity.--} Communication complexity deals with minimizing the amount of classical communication required to compute functions when inputs are distributed among different parties \cite{Yao1979, Kushilevitz1996, Kushilevitz1997}. Quantum communication complexity extends this concept by incorporating quantum resources. In the framework proposed by Yao \cite{Yao1993}, qubit transmission is permitted among the parties. Meanwhile, Cleve \& Buhrman’s entanglement-based model \cite{Cleve1997} involves sharing quantum entanglement and using classical bits for communication. The quantum advantage in this latter model arises from the nonlocal correlations inherent in quantum entanglement \cite{Buhrman2010}.

In a recent study \cite{Chakraborty2024}, the authors investigated a class of communication complexity problems, denoted $\mathrm{CC_n}$, involving $n$ distant senders and one computing server, to establish a scalable advantages of $(n+1)$-qubit GHZ state \(\ket{G_{n+1}}\). The specific case of \(n=2\) is of particular interest here. Two variants of the problem can be defined based on the promise conditions.
\begin{itemize}
\item[]$\mathrm{CC^E_2}$: Alice, Bob, and Charlie are given two-bit strings ${\bf x}=x^0x^1$, ${\bf y}=y^0y^1$, and ${\bf z}=z^0z^1$, sampled uniformly at random, satisfying the promise condition $x^0 \oplus y^0 \oplus z^0 = 0$. Charlie’s task is to compute the function
\begin{align}
f_E({\bf x},{\bf y},{\bf z}) = x^1 \oplus y^1 \oplus z^1 \oplus (x^0 \lor y^0 \lor z^0).
\end{align}
\item[]$\mathrm{CC^O_2}$: In this case, the strings ${\bf x}$, ${\bf y}$, and ${\bf z}$, satisfy the promise condition $x^0 \oplus y^0 \oplus z^0 = 1$, while Charlie must compute the function
\begin{align}
f_O({\bf x},{\bf y},{\bf z}) = x^1 \oplus y^1 \oplus z^1 \oplus (x^0 \land y^0 \land z^0).
\end{align}
\end{itemize}
Interesting, $\ket{G_3}$ state becomes efficient for achieving the tasks $\mathrm{CC^E_2}$ and $\mathrm{CC^O_2}$s.
\begin{theorem}\label{theo1}
The functions \( f_E \) and \( f_O \) can be computed exactly by Charlie with $1$ bit of communication from each of Alice and Bob, provided they share the state \(\ket{G_3}\).   
\end{theorem}
\begin{proof}
We detail the protocol for \(\mathrm{CC^O_2}\). A similar argument follows for \(\mathrm{CC^E_2}\) (see also \cite{Chakraborty2024}). Each party treats the first bit of their input string as the input for the odd-parity GHZ game and follows the protocol outlined in Table \ref{tab1}. Consequently, their local outcomes \(a, b, c\) satisfy the condition \(a \oplus b \oplus c = x^0 \land y^0 \land z^0\). Alice and Bob then respectively communicate \(c_A := a \oplus x^1\) and \(c_B := b \oplus y^1\) to Charlie. Charlie's final computation is: $c_A \oplus c_B \oplus c \oplus z^1 = x^1 \oplus y^1 \oplus z^1 \oplus (a \oplus b \oplus c) = x^1 \oplus y^1 \oplus z^1 \oplus (x^0 \land y^0 \land z^0) = f_O$. This completes the proof.
\end{proof}
As it turns out without $\ket{G_3}$ state Charlie fails to evaluate the functions \( f_E \) and \( f_O \) with limited classical communication from Alice and Bob.
\begin{table}[t!]
\centering
\begin{tabular}{|c||c|c|c|c|c|}
\hline
&\multicolumn{2}{|c|}{GHZ$_{\bf E}$} & \multicolumn{2}{|c|}{GHZ$_{\bf O}$}\\
\hline
&Input $=0$&Input $=1$&Input $=0$&Input $=1$\\
\hline
Alice& $\sigma_1$ & $\sigma_2$ & $\overline{\sigma}_2$ & $\sigma_1$\\
\hline
Bob& $\sigma_1$ & $\sigma_2$ & $\sigma_1$ & $\sigma_2$\\
\hline
Charlie& $\sigma_1$ & $\sigma_2$ & $\sigma_1$ & $\sigma_2$\\
\hline
\end{tabular}
\caption{Here, \(\sigma_1\) and \(\sigma_2\) correspond to Pauli-\(X\) and Pauli-\(Y\) measurements, respectively. For a Pauli measurement \(\sigma_{\hat{n}}\) along the direction \(\hat{n}\), \(\overline{\sigma}_{\hat{n}}\) denotes the measurement along the opposite direction \(-\hat{n}\). When observing measurement outcome \(+1\) or \(-1\), the parties declare their results as \(0\) or \(1\), respectively.}\vspace{-.3cm}
\label{tab1}
\end{table}
\begin{theorem}\label{theo2}
With $1$ bit of communication from each of Alice and Bob, Charlie cannot compute either \( f_E \) or \( f_O \) exactly, even with the assistance of classical shared randomness.
\end{theorem}
\begin{proof}
We present the proof for \( f_O \) only; a similar argument holds for \( f_E \) (see \cite{Chakraborty2024}). Charlie can compute \( f_O \) if and only if he can compute \( f'_O({\bf x}, {\bf y}, z^0) = x^1 \oplus y^1 \oplus (x^0 \land y^0 \land z^0) \). Given the promise condition \( x^0 \oplus y^0 \oplus z^0 = 1 \), the function \( f'_O \) can be rewritten as:
\begin{subequations}\label{f3}
\begin{align}
f^\prime_O({\bf x}, {\bf y}, 0) &= x^1 \oplus y^1, \hspace{2cm}\text{with} ~ x^0\neq y^0; \label{f30}\\
f^\prime_O({\bf x}, {\bf y}, 1) &=
x^1 \oplus y^1\oplus(x^0\land y^0), ~~ \text{with} ~ x^0=y^0.\label{f31}
\end{align}   
\end{subequations}
A general function $g:\{0,1\}^{\times 2} \mapsto \{0,1\}$ can be expressed as $g^{\alpha\beta\gamma\delta}(s^0, s^1) := \alpha s^0 \oplus \beta s^1 \oplus \gamma s^0 s^1 \oplus \delta$, with $\alpha, \beta, \gamma, \delta, s^0, s^1 \in \{0,1\}$. Consequently, the communication from Alice and Bob to Charlie can be represented as $c_i := \mathcal{E}_i({\bf k}_i) = \alpha_i k^0_i \oplus \beta_i k^1_i \oplus \gamma_i k^0_i k^1_i \oplus \delta_i$, with $i\in\{A,B\},~{\bf k}_A={\bf x},~\&~{\bf k}_B={\bf y}$. To compute the desired function in Eq. (\ref{f3}), Charlie applies a generic decoding function to the communication bits $c_A$ and $c_B$ received from Alice and Bob, respectively. Notably, Charlie's decoding can depend on the input $z^0$. Denoting the respective decoding functions as $\mathcal{D}_j(c_A, c_B) = \alpha^j_C c_A \oplus \beta^j_C c_B \oplus \gamma^j_C c_A c_B \oplus \delta^j_C$ for $j\in\{0,1\}$, exact computability of $f_O$ thus demands 
\begin{align}
\mathcal{D}_0=f^\prime_O({\bf x},{\bf y},0)~~~~\&~~~~\mathcal{D}_1=f^\prime_O({\bf x},{\bf y},1).   
\end{align}
A lengthy but straightforward calculation results in:
\begin{itemize}
\item[(i)] {For $z^0=0$, we have the restrictions $\beta_A=\beta_B=\beta^0_C=\alpha^0_C=1,~\gamma^0_C=0$, and $\alpha_A=\alpha_B$;}
\item[(ii)] {For $z^0 = 1$, we have the restrictions $\alpha^1_C = \beta^1_C =\beta_A=\beta_B = 1, \gamma^1_C = 0, \delta_A= \delta_B\oplus\delta^1_C$ and $\alpha_A \neq \alpha_B$.}
\end{itemize}
Since the restrictions in (i) and (ii) are not consistent with each other, therefore no consistent encodings $\mathcal{E}_A,\mathcal{E}_B$ and decodings $\mathcal{D}_0,\mathcal{D}_1$ exist that can evaluate the function $f_O$ exactly. Having no deterministic strategy their probabilistic mixtures also fail to compute the function $f_O$. This completes the proof.   
\end{proof}
Naturally, the question arises: how much classical resources are sufficient to evaluate the functions \( f_E \) and \( f_O \)? Our next proposition addresses this question.
\begin{proposition}\label{prop1}
Charlie can evaluate the functions \( f_E \) and \( f_O \) exactly with $2$ bits of communication from one sender (say Alice) and $1$ bit of communication from the other (Bob).
\end{proposition}
\begin{proof}
Using $2$ bits, Alice communicates her string \({\bf x}\) to Charlie, while Bob communicates \(y^1\). With knowledge of the promise condition, Charlie can deduce the first bit \(y^0\) of Bob's string and thus compute the corresponding function.
\end{proof}
Proposition \ref{prop1}, along with Theorems \ref{theo1} and \ref{theo2}, implies that the nonlocal correlation in GHZ$_{\bf E}$ (GHZ$_{\bf O}$) game provide a $1$-bit communication advantage in performing the \(\mathrm{CC^E_2}\) (\(\mathrm{CC^O_2}\)) task.   

{\it Randomized GHZ game.--} Consider a scenario where the Referee has a random variable \( r_1 \) that uniformly takes values from the set \(\{\mathrm{E}, \mathrm{O}\}\). Depending on the value of \( r_1 \), the Referee asks the parties to play either GHZ$_{\bf E}$ or GHZ$_{\bf O}$. However, value of the parity bit \( r_1 \) is revealed only to one party. This resulting game we termed as the RGHZ game. The Referee can introduce an additional layer of randomization with another independent random variable \( r_2 \), which uniformly takes values from the set \(\{\mathrm{A}, \mathrm{B}\}\). Depending on whether \( r_2 \) takes the value \(\mathrm{A}\) or \(\mathrm{B}\), the parity bit \( r_1 \) is revealed to either Alice or Bob, respectively \cite{Self}. Given the involvement of two random variables, this resulting game is termed as R$_2$GHZ. Denoting \(\mathrm{E} \equiv 0\) \& \(\mathrm{O} \equiv 1\) and \(\mathrm{A} \equiv 0\) \& \(\mathrm{B} \equiv 1\), the RGHZ game can be formally defined as follows:\\
$\Rightarrow$ Alice is provided with the inputs $x,r_2,r_1\overline{r}_2$; Bob the inputs $y,r_2,r_1r_2$; and Charlie the input $z$.\\
$\Rightarrow$ $x,y,z, r_1$ are uniformly sampled satisfying the promise condition, ${\bf P}_R:~x\oplus y\oplus z=r_1$.\\
$\Rightarrow$ The outcomes need to satisfy the winning condition $a\oplus b\oplus c=[(x\lor y\lor z)\land \overline{r}_1]\lor[(x\land y\land z)\land r_1]:=\Omega$.
\begin{lemma}\label{lemma1}
The R$_2$GHZ game can be perfectly won by Alice, Bob, and Charlie if they share the GHZ state \(\ket{G_3}\).
\end{lemma}
\begin{proof}
For $r_2=0$ case, Bob's and Charlie's strategies are independent of $r_1$, whereas Alice modifies her strategy depending on the random variable $r_1$ (see Table \ref{tab1}). For $r_2=1$, Alice and Bob interchange their respective strategies in Table \ref{tab1}.   
\end{proof}
A simple approach to quantifying the strength of nonlocal correlations is to compare their success rate in the corresponding Bell game to the best achievable classical success. For instance, the correlations obtained from the GHZ state \(\ket{G_3}\) achieve perfect success in the GHZ$_{\bf E}$ and GHZ$_{\bf O}$ games, while the optimal classical success in these games is only \( 3/4 \). In this sense, stronger nonlocality of the correlations yielding perfect success in R$_2$GHZ games is not apparent.
\begin{proposition}\label{prop2}
The optimal classical successes of the RGHZ and R$_2$GHZ games are both upper bounded by \( 3/4 \). Furthermore, this bound is achievable.
\end{proposition}
\begin{proof}
The upper bound \( 3/4 \) for the RGHZ game follows from the fact that a classical strategy achieving success greater than \( 3/4 \) would imply a success strictly greater than \( 3/4 \) in at least one of the GHZ$_{\bf E}$ or GHZ$_{\bf O}$ games, which is not possible. A similar argument applies to the R$_2$GHZ game.

To achieve this bound in RGHZ game, the parties can use the classical deterministic strategy \( a = \overline{x} \), \( b = y \), and \( c = z \). As shown in Table \ref{tab2}, this strategy yields the success \( 3/4 \). In the case of the R$_2$GHZ game, the bound of \( 3/4 \) is also achievable with the strategy \( a = \overline{x} \), \( b = y \), and \( c = z \) for \( r_2 = 0 \), and \( a = x \), \( b = \overline{y} \), and \( c = z \) for \( r_2 = 1 \). This completes the proof. 
\end{proof}
In a sense, Proposition \ref{prop2} is intuitive: how can the classical randomization  reveal a stronger form of nonlocality? Intriguingly, we will now demonstrate that this indeed amplifies the strength of nonlocality, which will be illustrated through the communication advantage obtained in a randomized variant, \(\mathrm{R_2CC_2}\), of the communication complexity task \(\mathrm{CC^E_2}/\mathrm{CC^O_2}\).
\begin{table}[t!]
\centering
\begin{tabular}{|c||c|c|c|c|c|}
\hline
$~~~r_1~~~$& $~~~x~y~z~~~$ & $~~~\Omega~~~$ & $~~~a~b~c~~~$ & $~a\oplus b\oplus c~$ & $~~\checkmark/\times~~$ \\ \hline\hline
\multirow{4}{*}{$0$} & $0~0~0$ & $0$ & $1~0~0$ & $1$ & $\times$ \\ \cline{2-6} 
                     & $0~1~1$ & $1$ & $1~1~1$ & $1$ & $\checkmark$ \\ \cline{2-6} 
                     & $1~0~1$ & $1$ & $0~0~1$ & $1$ & $\checkmark$ \\ \cline{2-6} 
                     & $1~1~0$ & $1$ & $0~1~0$ & $1$ & $\checkmark$ \\ \hline\hline
\multirow{4}{*}{$1$} & $0~0~1$ & $0$ & $1~0~1$ & $0$ & $\checkmark$ \\ \cline{2-6} 
                     & $0~1~0$ & $0$ & $1~1~0$ & $0$ & $\checkmark$ \\ \cline{2-6} 
                     & $1~0~0$ & $0$ & $0~0~0$ & $0$ & $\checkmark$ \\ \cline{2-6} 
                     & $1~1~1$ & $1$ & $0~1~1$ & $0$ & $\times$ \\ \hline
\end{tabular}
\caption{The deterministic strategy \( a = \overline{x} \), \( b = y \), and \( c = z \) satisfies the winning condition in $6$ out of the $8$ possible cases, thereby achieving the success rate \( 3/4 \).}\vspace{-.3cm}
\label{tab2}
\end{table}

{\it The task $\mathrm{R_2CC_2}$.--} Similarly to the \(\mathrm{CC^E_2}/\mathrm{CC^O_2}\) task, Alice, Bob, and Charlie are given two-bit strings \({\bf x} = x^0x^1\), \({\bf y} = y^0y^1\), and \({\bf z} = z^0z^1\). Additionally, Alice receives the bits \(r_2\) \& \(r_1 \overline{r}_2\), while Bob receives the bits \(r_2\) \& \(r_1 r_2\). The inputs are sampled satisfying the promise condition \(x^0 \oplus y^0 \oplus z^0 = r_1\). Charlie is tasked with evaluating the function: 
\begin{align}
&f_R({\bf x},{\bf y},{\bf z},r_1,r_2)=x^1\oplus y^1\oplus z^1\oplus\Omega_0,
\end{align}
where, $\Omega_0:=[(x^0\lor y^0\lor z^0)\land \overline{r}_1]\lor[(x^0\land y^0\land z^0)\land r_1]$.
\begin{theorem}\label{theo3}
The function \( f_R \) can be exactly evaluated by Charlie with 1 bit of communication from each of Alice and Bob, provided the parties share the GHZ state \(\ket{G_3}\).
\end{theorem}
\begin{proof}
The proof follows a reasoning analogous to that of Theorem \ref{theo1}. By treating the second bit of their respective strings as the inputs for the R$_2$GHZ game, Alice and Bob use the same strategy as in the R$_2$GHZ game. Alice and Bob then communicate \( c_A = a \oplus x^1 \) and \( c_B = b \oplus y^1 \) to Charlie. Charlie can then compute the target function exactly using the expression \( c_A \oplus c_B \oplus c \oplus z^1 \).
\end{proof}
An intriguing question is whether a classical strategy requiring three 1-bit channels, two with one party and one with the other, can exactly evaluate the function \( f_R \), similar to Proposition \ref{prop1}. Importantly, before the task begins, the three $1$-bit classical channels can be configured in two distinct ways: Configuration \({\bf C1}\): Two channels from Alice to Charlie and one channel from Bob to Charlie; Configuration \({\bf C2}\): One channel from Alice to Charlie and two channels from Bob to Charlie.
\begin{theorem}\label{theo4}
The function \( f_R \) cannot be evaluated with $3$ classical bit channels from Alice and Bob to Charlie, regardless of whether the channels are configured as \({\bf C1}\) or \({\bf C2}\).
\end{theorem}
\begin{proof}
Consider the channels configured as \({\bf C1}\), where Alice and Bob respectively shares two channels and one channel with Charlie. Let us analyze the scenario where \( r_2 = 1 \), meaning that the parity bit \( r_1 \) is available to Bob only. Now, consider an even stronger configuration \({\bf C1'}\), where Alice is physically located in Charlie's lab. In this \({\bf C1'}\) configuration, examine the sub-task with \( x^1 = z^1 = 0 \). For perfect evaluation of \( f_R \), in this sub-task Charlie and Alice together should be able to evaluate the function \( \tilde{f}_{x^0z^0} \), defined as:
\begin{align}
\left\{\!\begin{aligned}
\tilde{f}_{00}:=y^1\oplus(y^0\land \overline{r}_1),~\tilde{f}_{01}:=\tilde{f}^{10}_R=y^1\oplus\overline{r}_1,\\
\tilde{f}_{11}=y^1\oplus\overline{r}_1\oplus(y^0\land r_1)~~~~~~~~~~~
\end{aligned}\right\}.
\end{align}
The three functions \(\{\tilde{f}_{00}, \tilde{f}_{01}, \tilde{f}_{11}\}\), being mutually independent, cannot be perfectly evaluated by Charlie and Alice in the \({\bf C1'}\) configuration with only 1-bit communication from Bob. Consequently, the sub-task cannot be achieved in the \({\bf C1}\) configuration either. This completes the proof. 
\end{proof}

Our next result presents an exact protocol for the \(\mathrm{R_2CC_2}\) task using only classical resources.
\begin{proposition}\label{prop3}
The functions $f_R$ can be evaluated by Charlie exactly by sharing two 1-bit channels with each of Alice and Bob.   
\end{proposition}
\begin{proof}
Each of Alice and Bob communicates their respective bit strings \({\bf x}\) and \({\bf y}\) to Charlie through $2$-bit classical channel. From \(x^0\), \(y^0\), and \(z^0\), Charlie first determines the parity bit \(r_1\), and then evaluates the function \(f_R\).      
\end{proof}
Proposition \ref{prop3}, together with Theorems \ref{theo3} and \ref{theo4}, establishes that the \(\ket{G_3}\) state provides a $2$-bit communication advantage in the \(\mathrm{R_2CC_2}\) task. Since the correlation resulting from the conventional GHZ game offers only a $1$-bit communication advantage in the \(\mathrm{CC_2}\) task, the correlation emerging from the R$_2$GHZ game, therefore, demonstrates a kind of stronger nonlocality.
  
{\it Discussions.--} Beyond its foundational importance, quantum nonlocality has far-reaching applications, including device-independent quantum protocols that function without requiring knowledge of the internal workings of the devices involved \cite{Brunner2014, Scarani2012}. In the context of quantum communication complexity, the quantum advantage \cite{Buhrman1998, Cleve1999, Buhrman1999, Buhrman2001} is achieved by leveraging the nonlocal properties of quantum entanglement \cite{Buhrman2010}. Traditional communication complexity tasks \cite{Yao1979, Kushilevitz1996, Kushilevitz1997} involve inputs distributed across multiple servers, with one server tasked with evaluating a specific function. Our $\mathrm{R_2CC_2}$ task extends this setup by randomizing the function to be computed and distributing this information randomly among the sending servers, rather than the computing server. Remarkably, this extension demonstrates that correlations from a GHZ state can provide a greater advantage than those seen in the conventional GHZ paradox. For future work, it would be interesting to explore whether similar randomization techniques could be applied to other quantum pseudo-telepathy games, potentially uncovering even stronger advantages of quantum entanglement.

{\bf Acknowledgment:} KA acknowledges support from the CSIR project {$09/0575(19300)/2024$-EMR-I. SGN acknowledges support from the CSIR project $09/0575(15951)/2022$-EMR-I. MB acknowledges funding from the National Mission in Interdisciplinary Cyber-Physical systems from the Department of Science and Technology through the I-HUB Quantum Technology Foundation (Grant no: I-HUB/PDF/2021-22/008).


\begin{thebibliography}{41}
\bibitem [{Sel()}]{Self}%
	\BibitemOpen
	\bibfield  {title} {\bibinfo {title} {At this point, a bit more
			generalization is possible where \( r_2 \) uniformly takes the value
			\(\mathrm{A}\) or \(\mathrm{B}\) or \(\mathrm{C}\), and accordingly the
			parity bit \( r_1 \) is revealed to either {A}lice or {B}ob or {C}harlie.
			{H}owever, for our purpose we restrict \( r_2 \) to be a binary variable},\
	}\href@noop {} {\ }\BibitemShut {NoStop}%
\end{thebibliography}
%

\end{document}